\documentclass[12pt]{amsart}
\usepackage{amsmath, amssymb,latexsym, amsthm}
\usepackage{pstricks}
\topskip  \usepackage{hyperref}
\usepackage{bigfoot}

\numberwithin{equation}{section}
\newtheorem{theorem}{Theorem}[section]

\newtheorem{corollary}[theorem]{Corollary}

\newtheorem{conjecture}[theorem]{Conjecture}

\newtheorem{lemma}[theorem]{Lemma}

\begin{document}
\def\NN{\mathbb N}
\def\L{\mathcal{L}}
\def\W{\mathcal{W}}
\def\N{\mathcal{N}}
\def\P{\mathcal{P}}
\def\A{\text{P1}}
\def\B{\text{P2}}

\def\({{\rm{(}}}
\def\){{\rm{)}}}
\def\c{{\rm{CN}}}
\def\bn{{\rm{BN}}}
\def\p{{\mathbf p}}
\def\q{{\mathbf q}}
\def\e{{\mathbf e}}
\def\onev{{\mathbf 1}}
\def\ra{\rightarrow}

\def\+{\oplus}
\def\mn{\mbox{-}}
\def\newop#1{\expandafter\def\csname #1\endcsname{\mathop{\rm #1}\nolimits}}
\newcommand{\val}[2]{#1\begin{pspicture}(12pt,9pt)\psline[unit=4pt,fillcolor=black](0,2)(1,0)(2,0)(3,2)\end{pspicture}#2}
\newcommand{\arc}[2]{#1\begin{pspicture}(12pt,9pt)\pscurve[unit=4pt,fillcolor=black](0.2,0)(1.5,1.5)(2.8,0)\end{pspicture}#2}
\newcommand{\fl}[1]{\left\lfloor #1\right \rfloor}
\newcommand{\ceil}[1]{\left \lceil #1 \right\rceil}
\newcommand{\tokI}{\begin{pspicture}(12pt,9pt)\psframe[linewidth=1pt,framearc=.7,fillstyle=vlines*, fillcolor=green](1,0.3)\end{pspicture}}
\newcommand{\tokII}{\begin{pspicture}(12pt,9pt)\psframe[linewidth=1pt,framearc=.7,fillstyle=hlines*,
fillcolor=red](1,0.3)\end{pspicture}}
\newcommand{\tok}{\begin{pspicture}(12pt,9pt)\psframe[linewidth=1pt,framearc=.7,fillstyle=solid,
fillcolor=yellow](1,0.3)\end{pspicture}}
\newcommand{\es}{\begin{pspicture}(12pt,2pt)\psframe[linewidth=1pt,framearc=.4,fillstyle=solid,
fillcolor=blue](1,0.1)\end{pspicture}}
  \psset{xunit=1cm}
\psset{yunit=0.8cm}

\def\cb{\color{blue}}
\def\cg{\color[rgb]{0.1, 0.5, 0.2} }
\def\cp{\color[rgb]{0.5,0.2,0.5} }
\def\myred{\color{red}}
\newrgbcolor{purple}{0.7 0.2 0.7}
\newrgbcolor{orange}{1.0 0.5 0.0}
\newrgbcolor{mygreen}{0.1 0.5 0.2}

\makeatletter
\def\imod#1{\allowbreak\mkern10mu({\operator@font mod}\,\,#1)}
\makeatother

\pagenumbering{arabic}
\pagestyle{headings}
\def\sof{\hfill\rule{2mm}{2mm}}
\date{\today}
\title{Building Nim}
\maketitle

\begin{center}
{\bf Eric Duch\^{e}ne\footnote{Supported by the ANR-14-CE25-0006 project of the French National Research Agency}}\\
{\it Universit\'{e} Lyon 1, LIRIS, UMR5205, F-69622, France}\\
{\tt eric.duchene@univ-lyon1.fr}\\
\vskip 10pt
{\bf Matthieu Dufour}\\
{\it Dept. of Mathematics, Universit\'e du Qu\'ebec \`a Montr\'eal\\
Montr\'eal,  Qu\'ebec H3C 3P8, Canada}\\
{\tt dufour.matthieu@uqam.ca}\\
\vskip 10pt
{\bf Silvia Heubach}\\
{\it Dept. of Mathematics, California State University Los
Angeles\\
Los Angeles, CA 90032, USA}\\
{\tt sheubac@calstatela.edu}\\
\vskip 10pt
{\bf Urban Larsson\footnote{Supported by the Killam Trust}}\\
{\it Dalhousie University, Halifax, Canada}\\
{\tt urban.larsson@yahoo.se}\\
%{\tt }\\
\end{center}

%===========================================================================
\section*{Abstract}
The game of {\sc nim}, with its simple rules, its elegant solution and its historical importance is the  quintessence of a combinatorial game, which is why it led to so many generalizations and modifications. We present  a modification with a new spin:  {\sc building nim}. With given finite numbers of tokens and stacks, this two-player game is played in two stages (thus belonging to the same family of games as e.g. {\sc nine-men's morris}): first {\sc building}, where players alternate to put one token on one of the, initially empty, stacks until all tokens have been used. Then, the players play {\sc nim}. Of course, because the solution for the game of {\sc nim} is known, the goal of the player who starts {\sc nim} play is a placement of the tokens so that the Nim-sum of the stack heights at the end of {\sc building} is different from 0. This game is trivial if the total number of tokens is odd as the Nim-sum could never be 0, or if both the number of tokens and the number of stacks are even,  since a simple mimicking strategy results in a Nim-sum of 0 after each of the second player's moves. We present the solution for this game for some non-trivial cases and state a general conjecture.

\noindent{\bf Keywords}: Combinatorial game, Nim

\noindent{\bf 2010 Mathematics Subject Classification}: 91A46, 91A05
\thispagestyle{empty}
%===========================================================================
\section{Introduction}\label{Introduction}

The game of {\sc nim} is believed to have originated in China, but the exact origin is unknown. The earliest references in Europe are in the early 16th century. C.L. Bouton completely analyzed the game in 1901 \cite{cite:nim} and coined the name {\sc nim} (thought to be derived from the German word for ``to take"). The game is played on a finite number of stacks with a finite number of tokens. Two players alternate in moving, by selecting a stack and taking one or more tokens from that stack, until no further move is possible. A player unable to move loses (also called normal play). Figure~\ref{cite:nim} shows an example of a position.
\begin{center}
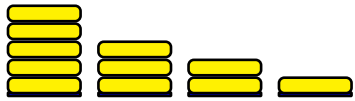
\begin{figure}[h]
\psset{xunit=1cm}
\psset{yunit=0.8cm}
\begin{pspicture}(4.8, 2)
\multirput(0,-0.2 )(1.2, 0){4}{\es}
\multirput(0,0)(0,0.3){5}{\tok}
\multirput(1.2,0)(0,0.3){3}{\tok}
\multirput(2.4,0)(0,0.3){2}{\tok}
\multirput(3.6,0)(0,0.3){1}{\tok}
%\multirput(3.6,0.35)(0,0.3){1}{\tokI}
\end{pspicture}
 \caption{The {\sc nim} position $(5, 3, 2, 1)$}\label{cite:nim}
\end{figure}
\end{center}

Note that there are many  variations on the basic game of {\sc nim}. A famous modification is the game of {\sc wythoff}~\cite{Wyt}; instead of removing tokens from a single stack, a player can also take the same number of tokens from two stacks. Another modification is the arrangement of stacks, such as in {\sc circular nim}~\cite{DH2013}. Other authors have considered {\sc nim} on graphs or simplicial simplexes~\cite{nimgraph}, and in \cite{Lar09} it is shown that for the game of {\sc imitation nim} a simple mimicking prevention rule in {\sc nim} gives the same $\P$-positions as {\sc blocking} {\sc wythoff} \cite{HeLa06}. A standard feature of many such games is that there are only two outcome classes; each game is either an $\N$- or a $\P$-position, that is, a position from which the current or previous player wins, respectively.

Here we present a new variation of {\sc nim} by introducing a {\sc building} stage before {\sc nim} begins. The game of {\sc building nim}, $\bn(n,\ell)$, is played with $n$ tokens on $\ell $ stacks in two stages:

\begin{itemize}
\item The first stage is {\sc building}: the two players take turns choosing one out of the $\ell$ stacks to place an unused token until all tokens have been used, resulting in a position of the form $s=(s_1,\ldots , s_\ell)$, where $s_i$ denotes the respective stack height, given in canonical form ordered from largest to smallest height (and some stacks may be empty);
\item The second stage is {\sc nim}: when all {\sc building} tokens have been used, the game of {\sc nim} starts from position $s$ with the player whose turn it is (that is, the player who did not place the last token in the {\sc building} stage).
\end{itemize}

Obviously, the winning strategy for {\sc building} is closely tied to that of {\sc nim}. The player who places the last token of {\sc building} would like to create a $\P$-position of {\sc nim}. Such a game having successive stages of play can be considered as a variation of {\em sequential compounds} of games~\cite{compound}, which consist in playing successive combinatorial games (with the objective of being the last player to move in the last game). The main difference here is that the {\sc building} stage is a different type of combinatorial game\footnote{It is a well-tempered (fixed-length) scoring game as defined by Johnson in \cite{Johns}. } and also that the opening of the second game depends non-trivially on the closure of the first.
%, because the player who starts {\sc nim} would like to start from an $\N$-position.

Similar to a {\sc building} position, a generic {\sc nim} position is represented by the vector of stack heights, $(s_1, s_2, \ldots, s_\ell)$. To describe the set of $\P$-positions for {\sc nim}, $\P$({\sc nim}), we define the {\it Nim-sum} $s_1\oplus s_2 \oplus \cdots \oplus s_\ell$, as obtained by translating the values into their binary representation and then adding them  without carry-over.

\begin{theorem}[Bouton \cite{cite:nim}]\label{losingnim}The $\P$-positions for {\sc nim} are those where the Nim-sum of the stack heights is 0, that is
$\P(\text{\sc nim}) =\{(s_1,s_2,\ldots,s_\ell)\mid s_1 \+ s_2 \+ \cdots \+ s_\ell=0\}.$
\end{theorem}

By this elegant formula, perfect {\sc nim} play boils down to a simple computation. Hence the {\sc building} stage is our only concern.
We denote by $\A$ the player who starts {\sc building nim}, and by $\B$ the second player. Hence, a phrase like ``$\A$ wins $\bn(n,\ell)$'' is equivalent to saying that this game is an $\N$-position.

We first state the trivial results.
\begin{theorem}[Easy cases] In the game $\bn(n,\ell)$, the following are true.
\begin{enumerate}
\item If $n$ is odd, then $\B$ wins. % Nim-sum can never be 0.
\item If both $n$ and $\ell$ are even, then $\B$ wins. %strategy stealing.
\end{enumerate}
\end{theorem}

\begin{proof} These statements follow directly from Theorem~\ref{losingnim}. When $n$ is odd, then the Nim-sum of the stack heights can never be zero, and therefore {\sc building} ends in a ({\sc nim}) $\N$-position. $\B$ starts {\sc nim}, and hence wins. If both $n$ and $\ell$ are even, then $\B$ can always mirror $\A$'s move in {\sc building}, resulting in pairs of stacks that have the same height. Since $a \+a=0$ for any $a$, {\sc building} ends in a $\P$-position for {\sc nim}, and therefore, since $\A$ starts {\sc nim}, again, she loses the game.
\end{proof}

This leaves the case when $n$ is even and $\ell$ is odd, and we will provide some explicit winning strategies. Specifically, we will prove that in the game of $\bn(n,\ell)$, with $n$ even and $\ell$ odd, the following holds.
\begin{enumerate}
\item If $\ell =3$, then $\B$ wins if and only if $n= 2^k-2$, for some positive integer~$k$;
\item If $\ell >3$, then $\B$ wins if $n\leqslant  \ell+3$;
\item If $\ell =5$, then $\A$ wins if and only if $n\geqslant 10$.
\end{enumerate}

Since the solutions build on particular ideas in the different cases, we will treat these cases as separate results, Theorems~\ref{s=3}, \ref{s>3}, and \ref{BT} respectively. Let us begin with some preliminary observations.

\section{Basic strategies and Nim-sum facts}

\begin{lemma} \label{strats} Consider an instance of $\bn(n, 3)$ for even $n\geqslant 1$ in {\sc building} play. At each turn $\A$ can force a position of the form $(y, x, x)$, with $y\geqslant x$, (Strategy I), while $\B$ at each turn can force a position $(z, x, y)$ with $z=x + y$ (Strategy II). \end{lemma}

\begin{proof}
In each case, we will show the claim by induction, considering two moves, one by each player, for the induction step. We will illustrate the relevant moves by showing the possible moves in a game tree. We use yellow tokens \tok \hspace{0.7cm} for the current position, a  green token \tokI \hspace{0.7cm} for a move made by $\A$, a red token \tokII \hspace{0.7cm} for a move made by $\B$, and \es \hspace{0.7cm} to indicate the stacks. In the game trees, we usually  show only one of two symmetric moves. Note that since $n$ is even, $\B$ makes the last {\sc building} move.

The first move for $\A$ is $(1,0,0)$ and  $\B$ can respond with $(1,1,0)$, corresponding to the desired form for the respective strategy. For the induction step for Strategy I, we need to show that if $\A$ has played to a position of the form $(y, x, x)$ with $y \geqslant x$, then no matter how $\B$ responds, $\A$ can counter to once more create such a position. Figure~\ref{strat1} shows the possible moves of $\B$ and the response  by $\A$. In each case, the resulting position is of the desired form.  Note that if $\A$ plays this strategy, then the final position after $\B$'s last move is of the form $(y, x, x)$ or $(y, x+1, x)$, with $y \geqslant x+1$.

Now look at the strategy for $\B$ and assume he leaves the stacks in the form $(z, x, y)$ with $z=x + y$ after his move. Figure~\ref{strat2} shows the possible moves of $\A$ and the response of $\B$.  Once more, it is possible to obtain a position of the desired type.
\end{proof}

 \begin{center}
 \begin{figure}[h]
\begin{pspicture}(9.4, 4.5)
 %\psset{xunit=.8cm}
%\psset{yunit=0.64cm}
% If smaller scale, then need to adjust width of picture to have it centered
\multirput(3,3.3 )(1.2, 0){3}{\es}
\multirput(3,3.5)(0,0.3){3}{\tok}
\multirput(4.2,3.5)(0,0.3){1}{\tok}
\multirput(5.4,3.5)(0,0.3){1}{\tok}
\psline{->}(3.2,2.8)(2.3,2)
\multirput(0,0.2 )(1.2, 0){3}{\es}
\multirput(0,0.4)(0,0.3){3}{\tok}
\multirput(1.2,0.4)(1.2,0){2}{\tok}
\multirput(0,1.3)(0,0.3){1}{\tokII}
\multirput(0,1.5)(0,0.3){1}{\tokI}
\psline{->}(6.2,2.8)(7.1,2)
\multirput(6,0.2 )(1.2, 0){3}{\es}
\multirput(6,0.4)(0,0.3){3}{\tok}
\multirput(7.2,0.4)(1.2,0){2}{\tok}
\multirput(7.2,0.7)(0,0.3){1}{\tokI}
\multirput(8.4,0.7)(0,0.3){1}{\tokII}
\end{pspicture}
 \caption{Sequence of moves from a position of type   $(y, x, x)$ to one of the same type two moves later.}\label{strat1}
\end{figure}
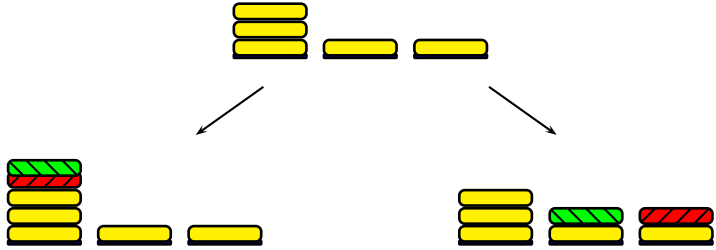
\end{center}

 \begin{center}
 \begin{figure}[h]
\begin{pspicture}(9.4, 4.8)
% \psset{xunit=.8cm}
%\psset{yunit=0.64cm}
% If smaller scale, then need to adjust width of picture to have it centered
\multirput(3,3.3 )(1.2, 0){3}{\es}
\multirput(3,3.5)(0,0.3){3}{\tok}
\multirput(4.2,3.5)(0,0.3){2}{\tok}
\multirput(5.4,3.5)(0,0.3){1}{\tok}
\psline{->}(3.2,2.8)(2.3,2)

\multirput(0,0.2 )(1.2, 0){3}{\es}
\multirput(0,0.4)(0,0.3){3}{\tok}
\multirput(0,1.3)(0,0.3){1}{\tokI}
\multirput(1.2,0.4)(0,0.3){2}{\tok}
\multirput(2.4,0.4)(0,0.3){1}{\tok}
\multirput(1.2,1)(0,0.3){1}{\tokII}

\psline{->}(6.2,2.8)(7.1,2)
\multirput(6,0.2 )(1.2, 0){3}{\es}
\multirput(6,0.4)(0,0.3){3}{\tok}
\multirput(6,1.3)(0,0.3){1}{\tokII}
\multirput(7.2,0.4)(1.2,0){2}{\tok}
\multirput(7.2,0.7)(0,0.3){1}{\tok}
\multirput(7.2,1)(0,0.3){1}{\tokI}

\end{pspicture}
\caption{Sequence of moves from a position of type $(z, x, y)$ with $z=x + y$ to one of the same type two moves later.}\label{strat2} \end{figure}
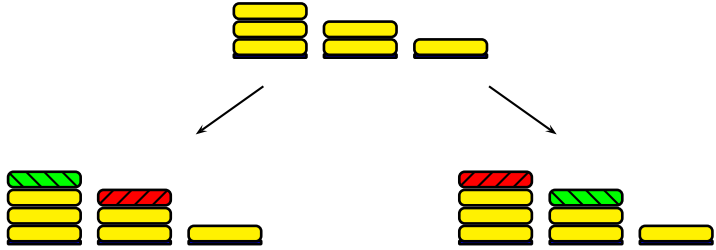
\end{center}

Let us present some basic results on the Nim-sum operator.

\begin{lemma}[Nim-sum facts] \label{DSfact} Let $x$, $y$, and $t_i, i=1, \ldots , \ell$, be integers. We have the following facts for the Nim-sum:
\begin{enumerate}
\item[(NS1)] $x\+y=0$ if and only if $x = y$.
\item[(NS2)] $x\+ y \leqslant  x+y$.
\item[(NS3)] $x\+ (x+1)=2^{k}-1$ for some $k\geqslant 1$.
\item[(NS4)] If $x=2^h-1$ for some $h$, then $x\+ (x+1)=2^{h+1}-1$; otherwise, $x\+ (x+1)< x$.
\item[(NS5)] $y=2^k-1$ for some $k \geqslant 1$ if and only if $x\+ (y-x) =y$ for $0 \leqslant  x \leqslant  y$.
\item[(NS6)] If $(s_1+1)  \+ s_2 \+ \cdots  \+ s_{\ell}=0$, then $s_1  \+ s_2 \+ \cdots \+ s_{\ell} = 2^k-1$ for some  $k \geqslant 1$.
\item[(NS7)] If $y > s_1+s_2+\cdots + s_\ell$, then $y \+s_1\+s_2\+\cdots \+ s_\ell >0.$
\end{enumerate}
\end{lemma}

\begin{proof}
In what follows we will use the notation $x=\ldots x_k x_{k-1} \ldots x_1 x_0$ for the binary expansion of $x= \sum_{i=0}^{\infty }  x_i 2^{i}$, where $x_i= 0$ or 1, with finitely many values of 1. If we want to put the emphasis on the fact that $x_i=1$ (or $0$), we will simply write $1_i$ (or $0_i$).
\begin{enumerate}
\item[(NS1):] If $x=y$, then $x$ and $y$ share the same unique binary expansion, that is,  $x_i =y_i$. As $x_i +x_i=0 \pmod 2$, we have that $x\+ y=0$.
If $x\+y=0$, then for each $i$, one has $x_i = y_i$, so $x=y$.
\item[(NS2):] If for every $i$, $x_i$ and $y_i$ are not both equal to $1$, then $x_i + y_i \pmod 2 = x_i+y_i$, so $x\+y=x+y$. If, on the contrary,  $x_i=y_i=1$ for some $i$,  then $x_i$ and $y_i$ cancel out in the Nim-sum, so  $x\+ y < x+y$.
\item[(NS3,4):] Let $h$ be the smallest index for which $x_h=0$. Then, $x=\ldots  x_{h+1} 0_h 1_{h-1} \ldots 1_1 1_0$, $x+1=\ldots  x_{h+1} 1_h 0_{h-1} \ldots 0_1 0_0$, and $x \+ (x+1) = 1_h 1_{h-1} \ldots 1_1 1_0 = 2^{h+1} - 1$,  which proves (NS3) (with $k=h+1$). Furthermore, if $x_j=0$ whenever $j>h$, then $x=2^h-1$ and $x\+(x+1)=x+(x+1)=2^{h+1}-1$. On the other hand, if $x_j=1$ for at least one $j>h$, then $x \geqslant 2^{h+1}$, and therefore, $x > x \+ (x+1) = 2^{h+1} - 1$, which proves (NS4).
\item[(NS5):] Suppose that $y=2^k-1$, so $y=1_{k-1} 1_{k-2}\ldots 1_1 1_0$.  For $0\leqslant  x\leqslant  y$ let $x_{k-1}x_{k-2}\ldots x_1 x_0$ be the binary expansion of $x$ and
$z_{k-1}z_{k-2}\ldots z_1 z_0$ be the binary expansion of $y-x$.  Then, for each $i=0,1,\ldots,k-1$, one has $x_i+z_i=1$, so $x\+(y-x)=x+(y-x)=y$.
Now suppose that $y \neq  2^k-1$ for any $k$. We will show that there is at least one pair  of integers $x, z$ such that $y=x+z$, but $x\+z \neq y$.
If $y \neq  2^k-1$ for any $k$, then the binary expansion of $y$ is not a string of consecutive ones, so there is at least one $0$ immediately to the right of a $1$. Let $h$ be the position of the rightmost such $0$, that is, $y=y_k y_{k-1} \ldots 1_{h+1} 0_h y_{h-1} \dots y_1 y_0$.
Define $x+1$ as the integer whose binary expansion is $y_k y_{k-1} \ldots 1_{h+1} 0_h 0_{h-1} \dots 0_1 0_0=2^{h+1}+\sum_{i=h+2}^{\infty }  y_i 2^{i}$, and $z-1$ as the integer whose binary expansion is given by $ 0_h y_{h-1} \dots y_1 y_0=\sum_{i=0}^{h-1 }  y_i 2^{i}$. Clearly, $(x+1)+(z-1)=x+z=y$
and the binary expansion of $x$ is given by $y_k y_{k-1} \ldots 0_{h+1} 1_h 1_{h-1} \dots 1_1 1_0$. As there are only $1$s in the $h+1$ rightmost digits of $x$, and because there is also at least one 1 in the (at most) $h+1$ digits of $z$, then at least one pair of $1$s will cancel out in the Nim-sum of $x$ and $z$, so $x\+z<x+z=y$,  completing the proof of (NS5). Here is a numerical  example that illustrates the proof. Let $y=25=11001_2$, so $h=2$ is the rightmost position of a $0$ following a $1$. Then $x+1=11000_2 = 24$ and $z-1=001_2=1$. That makes $x=23=10111_2$ and $z=010_2=2$  (using $h+1$ digits). The $1$s at position $h=1$ will cancel out in the Nim-sum, giving  $23\+2=21<25$.
\item[(NS6):] Let $y=s_2 \+ \cdots  \+ s_{\ell}$. Then $(s_1+1)  \+ s_2 \+ \cdots  \+ s_{\ell}=(s_1+1)\+y=0$ implies that $y = s_1+1$ by (NS1), and therefore, $s_1  \+ s_2 \+ \cdots  \+ s_{\ell}=s_1\+ y = s_1\+ (s_1+1)=2^k-1$ for some $k\geqslant 1$ by (NS3) and the proof is complete.
\item[(NS7):] Since $s_1\+s_2\+\cdots \+ s_k \leqslant   s_1+s_2+\cdots + s_k < y$, we have $y \+s_1\+s_2\+\cdots \+ s_k > 0$ by (NS1).
\end{enumerate}
\end{proof}
In the subsequent proofs, we will only use the ``only if" part of (NS5). An interesting corollary to (NS6) is currently not used in our proofs, but perhaps it has relevance to the solution of the general conjecture.

\begin{corollary} \label{Player1lm} $\A$ wins if her final move in {\sc building} play is to a position for which the Nim-sum of the stack sizes is not of the form $2^h-1$, for any positive integer $h$.
\end{corollary}
\begin{proof} Indeed, to win, $\B$ must finish with a Nim-sum of 0. Then, by (NS6), the position before his final move must have a Nim-sum of the form $2^h-1$, for some $h\geqslant 1$.
\end{proof}

Note that (NS6) is not true in the other direction, as for example, $2\+5=2^3-1$, but neither $3\+5$ nor $2\+6$ equals 0. Thus, Corollary~\ref{Player1lm} is also not an if and only if statement.

On the other hand, we will use (NS7) repeatedly in the proof for $\ell=5$ to conclude that $\A$ wins whenever she manages to build a stack that contains more than half of the tokens. Moreover, as we will see, if two stack heights are equal, then she wins if there is another stack with more than half of the tokens that are not in the matched stacks.

%%%%%%%%%%%%%%%%%%%%%%%%%%%%%%%%%%%%%%%%%%%%%%%%%%%%%%%%%%%%%%

We are now ready to state the main results. We first give the result for who wins on three stacks, as well as a general result that P2 wins when the number of tokens is at most three more than the number of stacks.

\section{Main results}

\begin{theorem} \label{s=3}
 If $n$ is even, then $\A$  wins $\bn(n,3)$ if and only if $n \ne 2^k-2$ for any $k$.
\end{theorem}

\begin{proof}   If $\A$ follows Strategy I, then {\sc building} ends in either  $(y,x,x)$, or  $(y,x+1,x)$ with $y \geqslant x+1$. In the first case, $\A$ wins as $y\, \+\, x \,\+ \, x   = y \, \+ \,0 \, > 0$ (so this is not a move $\B$ should make). In the second case, we need to distinguish between $x \ne 2^k-1$ and $x = 2^k-1$. If  $x \ne 2^k-1$, then
$y \+ (x+1) \+ x\ne 0$, as $x \+ (x+1)<x$ by (NS4) and $x \leqslant y$ by assumption. On the other hand, if $x = 2^k-1$, then we have that
$$ y\+ (x+1)\+ x = y \+ (2^{k+1}-1) = 0   \Leftrightarrow
 n =2^{k+2}-2.$$
 It remains to be shown that $\B$ can force a win in the case where $n =2^{k}-2$ for some $k$, no matter which strategy $\A$ employs.   Let $n =2^{k}-2$.  If $\B$ follows Strategy II, then the building phase ends in $(x+y,x,y)$. Since $ (x+y)+x + y = n = 2^k-2$, we have that $ x+y = 2^{k-1}-1$, and hence by (NS5), that $x \+  y = x+y$. This implies that
$(x+y) \+ x \+ y  = (x+y) \+ (x+y)=0$, a win for $\B$.
\end{proof}

For more than three stacks, the winner depends on the interplay between $n$ and $\ell$, as opposed to depending on the specific value of $n$ only.

\begin{theorem}\label{s>3}
P2 wins $\bn(n,\ell)$ for odd $\ell >3$ and even $n \leqslant  \ell+3$.
\end{theorem}

\begin{proof} We consider three cases, namely  $n \leqslant \ell-1$, $n = \ell+1$, and $n = \ell+3$. If $n \leqslant \ell-1$, then $\B$ can always mirror the move of $\A$ as there are more stacks than tokens. Pairs of stacks of equal height have a Nim-sum of zero, so the final position has Nim-sum zero in this case. If $n = \ell+1$ or $n = \ell+3$, then $\B$ plays the mirroring strategy but adjusts it as needed in the final two moves. To describe how the adjustment is made, we will describe a position as $(x_1, x_2, \ldots, x_\ell; r)$, where the first $\ell$ terms describe stack heights as before and the last term indicates the number of tokens (= number of moves) that remain to be played. Of course,  $r=n-x_1 - x_2 - \cdots - x_\ell$, but it will help for the clarity of the proof to emphasize the number of moves that remain.

Specifically, when $n=\ell+1$, the mirroring strategy does not work when $\A$ always starts a new stack, that is, if the position after the second to last move of $\A$ is $(1, \ldots,1, 1,1 ,1,0,0;3)$. $\B$ now adjusts his strategy and moves to position $(2, \ldots,1, 1,1, 1,0,0;2)$. Figure~\ref{n=s+1} shows  five stacks only (omitting the other pairs of matched stacks at height one) with the possible moves by $\A$ and the response by $\B$ to a position that has either matched stacks or a $1-2-3$ configuration, each resulting in a zero Nim-sum and a win for $\B$.

 \begin{figure}[htb]
 \begin{center}
 %\psset{xunit=.9cm}
%\psset{yunit=0.72cm}
 \psset{xunit=.8cm}
\psset{yunit=0.64cm}
\begin{pspicture}(14.4,5)
\multirput(4.8,3.9)(1.2, 0){5}{\es}
\multirput(4.8,4.2)(1.2,0){3}{\tok}
\multirput(4.8,4.5)(0,0.3){1}{\tok}

\psline{->}(7.7,3.3)(7.7,1.5)
\psline{->}(5,3.7)(3.5,3.3)
\psline{->}(10.4,3.7)(11.9,3.3)

\multirput(0,1.7 )(1.2, 0){5}{\es}
\multirput(0,2)(1.2,0){3}{\tok}
\multirput(0,2.3)(0,0.3){1}{\tok}
\multirput(3.6,2)(0,0.3){1}{\tokI}
\multirput(1.2,2.3)(0,0.3){1}{\tokII}

\multirput(4.8,0)(1.2, 0){5}{\es}
\multirput(4.8,0.3)(1.2,0){3}{\tok}
\multirput(4.8,0.6)(0,0.3){1}{\tok}
\multirput(8.4,0.3)(0,0.3){1}{\tokII}
\multirput(6,0.6)(0,0.3){1}{\tokI}

\multirput(9.6,1.8)(1.2, 0){5}{\es}
\multirput(9.6,2)(1.2,0){3}{\tok}
\multirput(9.6,2.3)(0,0.3){1}{\tok}
\multirput(9.6,2.6)(0,0.3){1}{\tokI}
\multirput(10.8,2.3)(0,0.3){1}{\tokII}

\end{pspicture}
\caption{Endgame when $n=\ell+1$.}\label{n=s+1}
\end{center}
\end{figure}
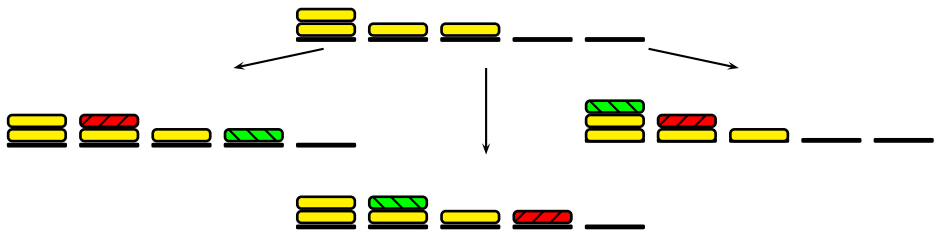

Next we look at the case $n=\ell+3$. Here, there are two positions where the mirroring strategy cannot be played until the end, namely $(2,2,1,1,\ldots, 1,0,0,0;4)$ or $(1,\ldots, 1,0;4)$. In the first case, the end game follows as in the case $n = \ell+1$ if $\A$ moves to  $(2,2,1,1,\ldots, 1,1,0,0;3)$, or by playing a mirroring strategy if $\A$ plays on a non-empty stack. In the second case, $\B$ adjusts his strategy as shown in Figure~\ref{n=s+3;a}  if $\A$ chooses to play on a non-empty stack in move $n-3$. Figure~\ref{n=s+3;b} shows the endgame if $\A$ plays on the empty stack in move $n-3$.

 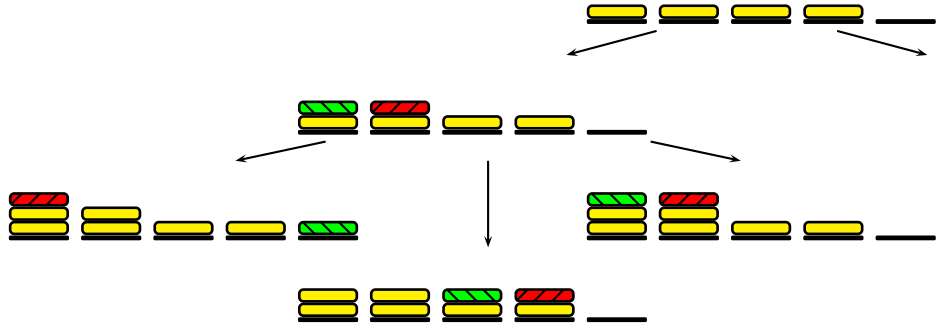
\begin{figure}[htb]
 \begin{center}
 \psset{xunit=.8cm}
\psset{yunit=0.64cm}
\begin{pspicture}(14.4, 7.5)

\multirput(9.6,6.2)(1.2, 0){5}{\es}
\multirput(9.6,6.5)(1.2,0){4}{\tok}

\psline{->}(13.5,6)(15,5.5)
\psline{->}(10.5,6)(9,5.5)

\multirput(4.8,3.9)(1.2, 0){5}{\es}
\multirput(4.8,4.2)(1.2,0){4}{\tok}
\multirput(4.8,4.5)(1.2,0){1}{\tokI}
\multirput(6,4.5)(1.2,0){1}{\tokII}

\psline{->}(7.7,3.3)(7.7,1.5)
\psline{->}(5,3.7)(3.5,3.3)
\psline{->}(10.4,3.7)(11.9,3.3)

\multirput(0,1.7 )(1.2, 0){5}{\es}
\multirput(0,2)(1.2,0){4}{\tok}
\multirput(0,2.3)(1.2,0){2}{\tok}
\multirput(4.8,2)(1.2,0){1}{\tokI}
\multirput(0,2.6)(1.2,0){1}{\tokII}

\multirput(4.8,0)(1.2, 0){5}{\es}
\multirput(4.8,0.3)(1.2,0){4}{\tok}
\multirput(4.8,0.6)(1.2,0){2}{\tok}
\multirput(7.2,0.6)(1.2,0){1}{\tokI}
\multirput(8.4,0.6)(1.2,0){1}{\tokII}

\multirput(9.6,1.7)(1.2, 0){5}{\es}
\multirput(9.6,2)(1.2,0){4}{\tok}
\multirput(9.6,2.3)(1.2,0){2}{\tok}
\multirput(9.6,2.6)(1.2,0){1}{\tokI}
\multirput(10.8,2.6)(1.2,0){1}{\tokII}

\end{pspicture}
\caption{Endgame when $n=\ell+3$ and $\A$ plays on a non-empty stack in move $n-3$.}\label{n=s+3;a}
\end{center}
\end{figure}

 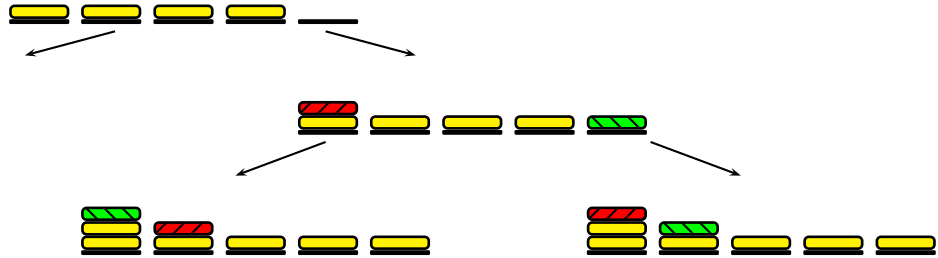
\begin{figure}[htb]
 \begin{center}
 \psset{xunit=.8cm}
\psset{yunit=0.64cm}
\begin{pspicture}(14.4, 7.5)

\multirput(0,6.2)(1.2, 0){5}{\es}
\multirput(0,6.5)(1.2,0){4}{\tok}

\psline{->}(1.5,6)(0,5.5)
\psline{->}(5,6)(6.5,5.5)

\multirput(4.8,3.9)(1.2, 0){5}{\es}
\multirput(4.8,4.2)(1.2,0){4}{\tok}
\multirput(4.8,4.5)(1.2,0){1}{\tokII}
\multirput(9.6,4.2)(1.2,0){1}{\tokI}

%\psline{->}(7.7,3.3)(7.7,1.5)
\psline{->}(5,3.7)(3.5,3)
\psline{->}(10.4,3.7)(11.9,3)

\multirput(1.2,1.4 )(1.2, 0){5}{\es}
\multirput(1.2,1.7)(1.2,0){5}{\tok}
\multirput(1.2,2)(1.2,0){1}{\tok}
\multirput(1.2,2.3)(1.2,0){1}{\tokI}
\multirput(2.4,2)(1.2,0){1}{\tokII}

%\multirput(4.8,0)(1.2, 0){5}{\es}
%\multirput(4.8,0.3)(1.2,0){4}{\tok}
%\multirput(4.8,0.6)(1.2,0){2}{\tok}
%\multirput(7.2,0.6)(1.2,0){1}{\tokI}
%\multirput(8.4,0.6)(1.2,0){1}{\tokII}

\multirput(9.6,1.4)(1.2, 0){5}{\es}
\multirput(9.6,1.7)(1.2,0){5}{\tok}
\multirput(9.6,2)(1.2,0){1}{\tok}
\multirput(9.6,2.3)(1.2,0){1}{\tokII}
\multirput(10.8,2.0)(1.2,0){1}{\tokI}

\end{pspicture}
\caption{Endgame when $n=\ell+3$ and $\A$ plays on the empty stack in move $n-3$.}\label{n=s+3;b}
\end{center}
\end{figure}

Once more the final positions consists of either matched stacks or a $1-2-3$ configuration.
\end{proof}

It may seem as if $\B$ might be able to adjust his strategy earlier and earlier and have a winning strategy also for larger values of $n$. However, one can check (by hand) that $\A$ has a winning strategy for $\bn(10,5)$ (see also Lemma~\ref{SpecialCases}) and some  other cases. Computer explorations lead to the following conjecture:

\begin{conjecture}\label{p1winsoften}
$\A$ wins $\bn(2n,\ell)$ if $2n > \ell + 3$.
\end{conjecture}

The proof for five stacks is more involved than that for three stacks, and it uses a number of ideas. Before we get into the technical details, we will state the result and discuss the main ideas. Theorem~\ref{BT} shows that Conjecture~\ref{p1winsoften} is true for $\ell =5$. It will be convenient to use $2n$ as the total number of tokens, that is,  the players play $n$ tokens each in the {\sc building} stage.

\begin{theorem}\label{BT}
$\A$ wins $\bn(2n, 5)$ if and only if $n \geqslant 5$.
\end{theorem}

The strategies of how $\A$ wins obviously vary depending on $\B$'s defense attempts, but parts of her ideas are independent of his responses. Item (NS7) of Lemma~\ref{DSfact} indicates that $\A$ wins whenever she manages to build a stack that contains more than half of the tokens. Moreover, if some stack heights are equal, then she wins if there is a stack with a height that is more than half of the tokens that are not in the {\it matched} stacks. So one of the general strategies for $\A$ will be to {\it play high}. This {\it height strategy} consists of playing on the tallest stack (possibly disregarding a pair of matched stacks). Sometimes the height strategy is not appropriate. In such situations, $\A$ wants to avoid helping $\B$ match up a tall stack, typically one with a height that is a  power of two, and therefore {\it plays low}. The {\it low strategy} consists of always playing on the minimal stack. Note that Strategy~I played on three stacks is a combination of the high and low strategies, selected in response to the various moves by $\B$. A nontrivial variation of this will be true also in the case of five stacks. %, where $\A$ decides whether to play high or low depending on the moves by $\B$. (However powerful these ideas may seem, we already know that such strategies need to be modified further to cover cases not already in this paper.)
At the core of the proof of Theorem~\ref{BT} is the idea that $\A$ can win by playing high, playing low, or by using the winning strategy from a game with fewer tokens for a game with more tokens, thus allowing us to do an inductive proof. We let the computer verify that $\A$ can win $\bn(2n, 5)$ for several small $n \geqslant 5$, and then proceed to prove that $\A$ can win all games for larger values of $n$.

Powers of 2 will play a pivotal role for the players' {\sc building} strategies. Hence we introduce the following terminology. Let $\pi$ be a given power of 2 strictly smaller than the number $n$ of tokens of each player. A game is strategically played in two {\sc building} phases:
\begin{itemize}
\item the $\pi$-phase: both players play their first $\pi<n$ tokens
\item the $\delta$-phase: both players play their remaining $\delta=\delta(n,\pi)=n-\pi>0$ tokens.
\end{itemize}

A special case is when the $\pi$-phase results in two matched stacks, and this is the instance where $\A$ wants to play a winning strategy for the $2\delta$ remaining tokens ``on top" of these two stacks if such a strategy exists. Figure~\ref{winontop} illustrates this idea. %Henceforth, tokens played by $\A$ are green and those played by $\B$ are red; yellow tokens are.

\begin{figure}[htb]
\begin{center}
\psset{xunit=1cm}
\psset{yunit=0.8cm}
\begin{pspicture}(13,3)
\multirput(0,-0.2 )(1.2, 0){5}{\es}
\multirput(0,0)(0,0.3){4}{\tok}
\multirput(1.2,0)(0,0.3){2}{\tok}
\multirput(2.4,0)(0,0.3){2}{\tok}
\multirput(3.6,0)(0,0.3){1}{\tok}
\multirput(4.8,0)(0,0.3){1}{\tok}
\psline[linewidth=1.4pt]{->}(6,1.0)(6.8,1.0)
\psline{|-|}(7.5,-0.2)(7.5,1)
\rput(7.2,.35){$2^{k}$}
\multirput(8,-0.2 )(1.2, 0){5}{\es}
\multirput(8,0)(0,0.3){4}{\tokI}
\multirput(9.2,0)(0,0.3){4}{\tokII}
\multirput(8,1.2)(0,0.3){4}{\tok}
\multirput(9.2,1.2)(0,0.3){2}{\tok}
\multirput(10.4,0)(0,0.3){2}{\tok}
\multirput(11.6,0)(0,0.3){1}{\tok}
\multirput(12.8,0)(0,0.3){1}{\tok}
%\multirput(3.6,0.35)(0,0.3){1}{\tokI}
\end{pspicture}
\end{center}
\caption{Reusing a winning strategy: A winning strategy for $2\delta$ is played ``on top" of two stacks of height $2^k$. }\label{winontop}
\end{figure}
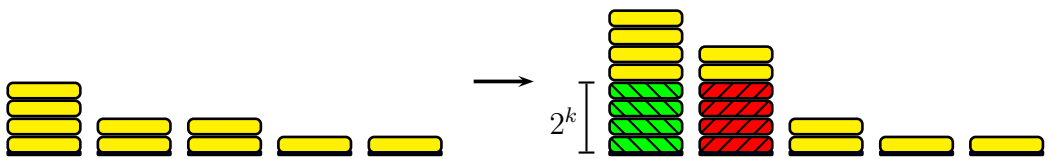

Consider an odd integer $\ell \geqslant 3$ and a positive integer $n$. Let $\delta+\pi=n$, where $\pi$ is a power of 2 and where $0<\delta < 2\pi$. The following lemma shows that if $\A$ wins $\bn(2\delta, \ell)$ then $\A$ wins $\bn(2n, \ell)$, if the players have built up two matching stacks in the $\pi$-phase.

\begin{lemma}\label{DS8}
Let $\pi$ be a power of 2 and let $\ell \geqslant 3$ be odd. Further, let $x_1,x_2,\ldots , x_\ell$ be integers with non-zero Nim-sum, but  $(x_1+ \pi) \+ (x_2+\pi)\+ x_3 \+ \cdots  \+ x_\ell =0$. Then
$$x_1+x_2+\cdots +  x_\ell \geqslant 4\pi.$$
\end{lemma}

\begin{proof}
Since the Nim-sum of the $x_i$ is non-zero, and the addition of $\pi=2^k$ cannot affect the Nim-sum of the coefficients of $2^r$ with $r < k$, these coefficients already must have a Nim-sum of zero. Therefore, we can disregard those coefficients in the argument, which amounts to proving the result for $k=0$. Furthermore, without loss of generality one can assume that there are only three stacks, as the stacks $x_3$ to $x_\ell$ can be replaced by a stack of height $x_3'=x_3 \+ x_4 \+ \cdots  \+ x_\ell$, using that $x_3 + x_4 + \cdots  + x_\ell \geqslant x_3'$.
 So it suffices to prove the following simpler fact:\\
If
\begin{align}x_1 \+ x_2 \+ x_3 >0&\label{(1)}\\
(x_1+1) \+ (x_2+1 )\+ x_3 =0 &\label{(2)}
\end{align}
then $x_1+ x_2 +x_3 \geqslant 4$.

Suppose that $x_1+x_2+x_3<4$. The smallest configuration for which a Nim-sum of three pairwise distinct numbers is 0 is $(1,2,3)$. Hence two of the terms in (\ref{(2)}) must be equal and the third must equal $0$. Notice that $x_1=x_2$ is impossible if both (\ref{(1)}) and (\ref{(2)}) are satisfied. Also $x_i+1=x_3$ forces $x_{3-i}=-1$ for $i=1,2$ which is impossible.
\end{proof}

Note that this result does not depend on the particular odd number of stacks. But it does depend on having a winning strategy for smaller games. By Theorem~\ref{s>3}, $\B$ wins for the smallest values of $n$. Therefore we will depend on knowledge of specific `initial' cases for which $\A$ wins. We can prove manually that $\A$ wins for the necessary initial cases for five stacks, but we have not yet found any general strategy. Moreover,  the next lemma probably generalizes to more than five odd stacks, but we do not yet have a general proof. Lemma~\ref{2^k-2} considers the cases $2n=2^k-2$ for $k> 4$ (games for which $\A$ loses when playing on  three stacks). In what follows, we will say that a stack has a {\it $k$-component} when the coefficient of $2^k$ in that stack height is not zero. Clearly, whenever {\sc nim} begins with an odd number of $k$-components, for some $k$, then $\A$ wins.

\begin{lemma}\label{2^k-2}
For $k>4$,  $\A$ wins $\bn(2^k-2,5)$.
\end{lemma}

\begin{proof}
We show that by playing low,  $\A$ can force an odd number of $k$-components, for some $k$, when {\sc nim}-play starts. For ease of describing the argument, we say that a token belongs to the {\it bottom layers} if the number of tokens below it is strictly smaller than $2^{k-3}$ (see Figure~\ref{bottom}).  Note that each player has $2^{k-1}-1$ tokens to play.

Let us first assume that $\B$ contributes at least $2^{k-3}+1$ tokens to the bottom layers. In this case, $\A$ will be able to ensure that  the bottom layers are completely filled by playing low.
No matter how many tokens $\B$ contributes beyond the $2^{k-3}+1$ tokens, there are a total of  $2n-5\cdot 2^{k-3}=2^k-2-5\cdot 2^{k-3}=3\cdot 2^{k-3}-2$ tokens in the upper layers.  How can they be distributed?   There are three distinct possibilities how the total configuration of {\sc building} can end:
\begin{itemize}
\item If  $s_1\geqslant s_2 \geqslant 2^{k-2}$, then there are three $(k-3)$-components, since $s_3<2^{k-2}$;
\item If  $s_1 \geqslant 2^{k-2}$ and $s_2 < 2^{k-2}$, then there is an unmatched $(k-2)$-component;
\item If  $s_1< 2^{k-2}$ then there are five $(k-3)$-components.
\end{itemize}
In all instances, {\sc building} play ends in a nonzero Nim-sum.

Suppose next that  $\B$ contributes at most $2^{k-3}$ tokens to the bottom layers.
Since $\A$ plays low, $\B$ has to contribute at least $2^{k-3}-1$ tokens to the bottom layers, so there are just two different possibilities. Let's first consider the case when  $\B$ contributes exactly this minimum number of tokens. In this case,  it will take until $\A$ has made $4(2^{k-3}-1)+2= 2^{k-1}-2$ moves (see Figure~\ref{bottom}) before $\B$ can play anywhere but on the first stack. At this point, $\B$ will have made one fewer move, so $\B$ has two moves left, and $\A$ has one move left.
 \begin{center}
 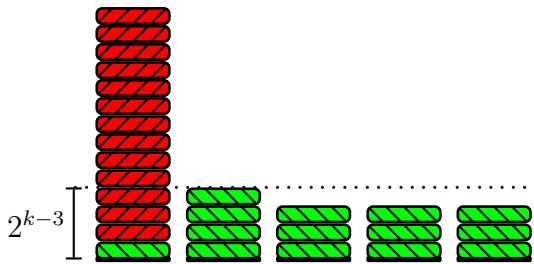
\begin{figure}[h]
  \psset{xunit=1cm}
\psset{yunit=0.8cm}
\begin{pspicture}(6, 4.3)
\psline{|-|}(-0.5,-0.2)(-0.5,1)
\rput(-1,.4){$2^{k-3}$}
\multirput(0,-0.2 )(1.2, 0){5}{\es}
\multirput(0,0)(0,0.3){1}{\tokI}
\multirput(0,0.3)(0,0.3){13}{\tokII}
\multirput(1.2,0)(0,0.3){4}{\tokI}
\multirput(2.4,0)(0,0.3){3}{\tokI}
\multirput(3.6,0)(0,0.3){3}{\tokI}
\multirput(4.8,0)(0,0.3){3}{\tokI}
\psline[linestyle=dotted, linewidth=1.2pt](-0.5,1.0)(5.5,1.0)
%\multirput(3.6,0.35)(0,0.3){1}{\tokI}
\end{pspicture}
 \caption{$\A$ plays the bottom layers whereas $\B$ plays on stack 1 when $k=5$.}\label{bottom}
\end{figure}
\end{center}
With  the initial token from $\A$, the first stack's height is $s_1=4 \cdot 2^{k-3}-2=2^{k-1}-2$, as shown in Figure~\ref{bottom} for $k=5$. Because $k>4$ by assumption,   $s_1$ contains a $(k-2)$-component. To match this $(k-2)$-component on stack ~2, a total of $2^{k-2}-2^{k-3}=2^{k-3}\geqslant 4$ tokens are needed. In this case, only three tokens remain to be played by the two players together, so $\B$ cannot build up a $(k-2)$-component in stack~2 to match the one of stack~1, irrespective of how $\A$ plays.

Now we look at the remaining case when $\B$ contributes $2^{k-3}$ tokens to the bottom layers. In this case, $\B$ can play on the upper layer of stack~2 earlier, after $\A$ has played $4(2^{k-3}-2)+2=2^{k-1}-6$ tokens. At this point, stack~1 contains $2^{k-1}-6$ tokens, and since $k \geqslant 5$, it contains a $(k-2)$-component. After the move by $\B$ shown in Figure~\ref{bottom2}, each player has exactly $(2^{k-1}-1)-(2^{k-1}-6)=5$  moves left, irrespective of $k$.

 \begin{center}
 \begin{figure}[h]
  \psset{xunit=1cm}
\psset{yunit=0.8cm}
\begin{pspicture}(6, 3.3)
\psline{|-|}(-0.5,-0.2)(-0.5,1)
\rput(-1,.4){$2^{k-3}$}
\multirput(0,-0.2 )(1.2, 0){5}{\es}
\multirput(0,0)(0,0.3){1}{\tokI}
\multirput(0,0.3)(0,0.3){9}{\tokII}
\multirput(1.2,0)(0,0.3){3}{\tokI}
\multirput(1.2,0.9)(0,0.3){1}{\tokII}
\multirput(2.4,0)(0,0.3){2}{\tokI}
\multirput(3.6,0)(0,0.3){2}{\tokI}
\multirput(4.8,0)(0,0.3){2}{\tokI}
\psline[linestyle=dotted, linewidth=1.2pt](-0.5,1.0)(5.5,1.0)
%\multirput(3.6,0.35)(0,0.3){1}{\tokI}
\end{pspicture}
 \caption{This is the case where $\B$ contributes $2^{k-3}$ tokens to the bottom layers. When $k=5$, if he plays strategically on stack 2, he has enough remaining tokens to be able to match the $({k-2})$-component in $s_1$.}\label{bottom2}
\end{figure}
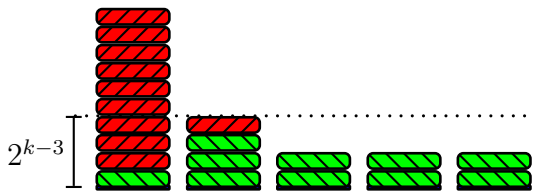
\end{center}

To build up a matching $(k-2)$-component in the second stack, $2^{k-3}$ tokens are needed (see above). If $k>5$, then $\B$ does not have enough tokens to do so on his own, and $\A$ can avoid helping in the build-up by playing high. If $k=5$, then $\B$ has enough remaining tokens to match the $3$-component by playing 4 of his tokens on stack ~2. $\A$ can counter by playing her $5$ tokens on stack ~3 (and her stack will never be higher than his). The position reached before the last token is placed by $\B$ is $(10,8,7,2,2)$, and no matter where he plays, there will be an unmatched component in stack~3 - either a  $2$-component if he does not play on stack 3 or a $3$-component if he does. If $\B$ does not build up the $3$-component in stack~2, then that component will be unmatched. Finally, if $\B$ waits longer to play his $4$th token in the bottom layers, then he will not have enough tokens left to match the $3$-component in stack 2. In all  cases, $\A$ wins.
\end{proof}

The proof of Theorem~\ref{BT} will make clear why we need to check the following cases.

\begin{lemma}\label{SpecialCases}
$\A$ has a winning strategy for $\bn(2n,5)$ for $n=5, \ldots, 12$.
\end{lemma}

\begin{proof}
These cases have been checked by a computer program, using the following natural algorithm derived from the definition of $\P$- and $\N$-positions: given a number $2n$ of tokens,
we start by computing the outcomes of the positions $(x_1,\ldots,x_5;0)$. Clearly those which satisfy $x_1\oplus\cdots\oplus x_5=0$ are $\P$, and the other ones are $\N$. Now, for all $\xi$ from $1$ to $2n$ we compute the outcomes of all positions $(x_1,\ldots,x_5;\xi)$ with $x_1+\cdots+x_5=2n-\xi$ as follows: a position $(x_1,\ldots,x_5;\xi)$ is $\N$ if at least one of its options is $\P$, otherwise it is $\P$. Note that each position $(x_1,\ldots,x_5;\xi)$ admits five options, namely
$$
\{(x_1+t_1,x_2+t_2,x_3+t_3,x_4+t_4,x_5+t_5;\xi-1)\;:\;t_i\in\{0,1\}, \sum_{i=1}^{5}t_i=1\}.
$$
A computation of the outcomes of positions $(0,\ldots,0;2n)$ for $n=5, \ldots, 12$ shows they are $\N$.
\end{proof}

We are now ready to prove Theorem~\ref{BT}.

%\subsection{Proof of main theorem for $\ell =5$}\label{sec:main}
\begin{proof}[Proof of Theorem~\ref{BT}]
We must prove that $\A$ wins $\bn(2n, 5)$ if and only if $n \geqslant 5$. By Theorem~\ref{s>3}, $\B$ wins if $n<5$. For $n\geqslant 5$ a given integer, let $k(n)$ be the unique integer such that $2^{k(n)-1}<n\leqslant 2^{k(n)}$ and let $p(n)=2^{k(n)-1}$ be the largest power of $2$ strictly less than $n$.  For convenience, we will replace $k(n)$ and $p(n)$ by $k$ and $p$  respectively, when the context is clear. We will proceed by induction on $n$. According to Lemma~\ref{SpecialCases}, $\A$ wins for all $5\leqslant n\leqslant 12$, and in particular when $k(n)=3$. Let $n>12$ and assume that $\A$ wins for all $m$ with $3\leqslant k(m)<k(n)$.

If $n=2^k-1$, Lemma~\ref{2^k-2} applies and $\A$ wins. If  $n\ne 2^k-1$, then
 one of $\A$'s strategies will be to reuse the winning strategy of a smaller game if $\B$ matches a power of two, $\pi$, in the $\pi$-phase; see the respective cases (ii) below. By Lemma~\ref{DS8}, she has to be careful to choose an appropriate $\pi$, because $\bn(2\delta, 5)$ is not a first player win for $\delta \leqslant 4$. Therefore, we consider two cases.
%\newline

{\bf Case 1:} $n-p>4$.

Here $\A$ chooses $\pi=p$, so $\delta=n-p > 4$ and $k(\delta) \geqslant 3$. In the $\pi$-phase (which consists of playing $p$ tokens each) $\A$ plays high, independent of $\B$'s responses. Then $\A$ adjusts her strategy for the $\delta$-phase depending on the play of $\B$ in the $\pi$-phase:

(i) If $\B$ played at least one of his $\pi$-phase tokens on the tallest stack, then $\A$ continues to play high in the $\delta$-phase.  At the end of {\sc building}, the maximum stack height will be larger than the sum of the other stacks, and $\A$ wins by (NS7) (Lemma~\ref{DSfact}).

(ii) If $\B$ has matched $\A$'s play on stack 2, then by the contraposition of Lemma~\ref{DS8}, $\A$ wins $\bn (2n,5)$ as she has a winning strategy for $\bn (2\delta, 5)$. Indeed, the conditions of Lemma~\ref{DS8} are fulfilled since $\delta<2\pi$, and by induction hypothesis, $\A$ can win $\bn(2\delta,5)$ since $3\leqslant k(\delta)<k(n)$.

(iii) The remaining case is that $\B$ neither played on the tallest stack nor matched play by $\A$ on the second stack. This means that $s_1=\pi$, $s_2<\pi$, and $s_3>0$. The case $n=2^{k}=2\pi$ is trivial since $\A$ continues to play high on stack 1 in the $\delta$-phase and wins, because there will be an unmatched $k$-component in stack 1 when {\sc nim} starts. Since $n\ne 2^k-1$, we may assume that $n<2^k-1\leqslant 2\pi-1$, that is
\begin{align}\label{<delta}
\delta <\pi-1.
\end{align}
%Then $\A$ uses one of two strategies depending on the moves by $\B$.

In this case, $\A$ adjusts her moves in response to $\B$'s play. The strategy of $\A$ hinges on whether $\B$ will be able to build stack 2 to a height of $\pi$ to match stack 1 or not. To prove that $\A$ has a winning strategy, we keep track of the stack heights after each pair of moves. After $\xi$ moves have been played by both players in the $\delta$-phase, starting with $\A$, each player has $\delta-\xi$ tokens to play, and we denote the number of tokens on the $i^{\text{th}}$ stack by $s_i(\xi)$. We discuss two ways in which $\A$ can win. %wins, keeping track of the stack heights after each pair of moves by the two players.

If at the end of {\sc building}, $\B$ has failed to match stack $s_1$ (by not having built up stack $s_2$ to a level of $\pi$ tokens), then there will be an unmatched $(k-1)$-component, so $\A$ wins. We claim that this type of position will be reachable for $\A$ if, after $\xi$ moves by each player in the $\delta$-phase, the number of tokens to be played by P2 is insufficient to cover the gap between $s_2(\xi)$ and $\pi$ (see Figure~\ref{single_pi}).

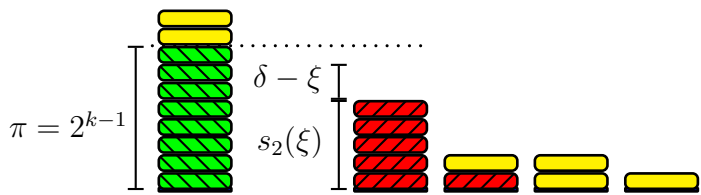
\begin{figure}[htb]
\begin{center}
\psset{xunit=1cm}
\psset{yunit=0.8cm}
\begin{pspicture}(6,3)
%\multirput(0,0)(0,0.3){4}{\tok}
%\multirput(1.2,0)(0,0.3){4}{\tok}
%\multirput(0,-0.2 )(1.2, 0){5}{\es}
%\multirput(0,1.2)(0,0.3){4}{\tok}
%\multirput(1.2,1.2)(0,0.3){2}{\tok}
%\multirput(2.4,0)(0,0.3){2}{\tok}
%\multirput(3.6,0)(0,0.3){1}{\tok}
%\multirput(4.8,0)(0,0.3){1}{\tok}
%\psline[linewidth=1.4pt]{->}(6,1.0)(6.8,1.0)
\psline{|-|}(-.7,-0.2)(-.7,2.2)
\rput(-1.6,.9){$\pi=2^{k-1}$}
\multirput(-0.2,-0.2 )(1.2, 0){1}{\es}

\multirput(-.2,0)(0,0.3){8}{\tokI}
\multirput(-.2,2.4)(0,0.3){2}{\tok}

\multirput(2.4,-0.2 )(1.2, 0){4}{\es}
\multirput(2.4,0)(0,0.3){5}{\tokII}
\multirput(3.6,0)(0,0.3){1}{\tokII}
\multirput(3.6,0.3)(0,0.3){1}{\tok}
\multirput(4.8,0)(0,0.3){2}{\tok}
\multirput(6,0)(0,0.3){1}{\tok}

\psline{|-|}(2,1.3)(2,1.9)
\psline{|-|}(2,-0.2)(2,1.3)

\rput(1.32,1.65){$\delta-\xi$}
\rput(1.35,.6){$s_2(\xi)$}
\psline[linestyle=dotted, linewidth=1.2pt](-0.5,2.2)(3.1,2.2)
%\multirput(3.6,0.35)(0,0.3){1}{\tokI}
\end{pspicture}
\end{center}
\caption{A $\A$ strategy for achieving a single $(k-1)$-component.}\label{single_pi}
\end{figure}

Suppose that
\begin{equation*}
s_2(\xi)+(\delta-\xi)<\pi
\end {equation*}
or equivalently,
\begin{equation}\label{Eq1}
s_2(\xi)-\xi < \pi-\delta
\end {equation}
holds at some stage in the $\delta$-phase. The claim is that $\A$ can ensure it still holds for the rest of {\sc building}, by always playing high. Then
\begin{align}\label{Eq2}
s_2(\xi+1)-(\xi +1)\leqslant s_2(\xi)+1-(\xi+1) < \pi-\delta,
\end{align}
and inequality~\eqref{Eq1} will hold also for $\xi +1$.%, unless, at some point, $\A$ is forced to play on stack 2. The only instance when this is the case is when $s_2(\xi )=s_5(\xi )$, that is, when all but the tallest stack have the same height. However, this state can only persist for possibly the first move. Once the stacks $s_2$ through $s_5$ are not all equal, $\A$ can ensure that they remain so until the end of {\sc building}. Namely, in case $\B$ does not play on $s_2$ at some move of the $\delta$-phase, then $\A$ can do so, without contradicting inequality~\eqref{Eq1}. (Three stacks will remain smaller than stack 2 if at least one of the players keep playing on stack 2, which is smaller than stack 1.)

On the other hand, if P2 has enough tokens to complete the second stack to size $\pi$, then $\A$ wins if she builds up stack $s_3$ to a height of more than $\delta$ tokens. Indeed, at the end of the $\delta$-phase we would have a position of the form $(\pi+x'_1,\pi+x'_2,s_3,s_4,s_5)$ and then she wins once more by (NS7) of Lemma~\ref{DSfact}, now applied to $y=s_3$ and $x'_1, x'_2, s_4, s_5$ (see~Figure~\ref{NS7}).

\begin{figure}[htb]\label{fig:10}
\begin{center}
\psset{xunit=1cm}
\psset{yunit=0.8cm}
\begin{pspicture}(6,3)
%\multirput(0,0)(0,0.3){4}{\tok}
%\multirput(1.2,0)(0,0.3){4}{\tok}
%\multirput(0,-0.2 )(1.2, 0){5}{\es}
%\multirput(0,1.2)(0,0.3){4}{\tok}
%\multirput(1.2,1.2)(0,0.3){2}{\tok}
%\multirput(2.4,0)(0,0.3){2}{\tok}
%\multirput(3.6,0)(0,0.3){1}{\tok}
%\multirput(4.8,0)(0,0.3){1}{\tok}
%\psline[linewidth=1.4pt]{->}(6,1.0)(6.8,1.0)
\psline{|-|}(-.7,-0.2)(-.7,2.2)
\rput(-1.6,.9){$\pi=2^{k-1}$}
\multirput(-0.2,-0.2 )(1.2, 0){1}{\es}

\multirput(-.2,0)(0,0.3){8}{\tokI}
\multirput(-.2,2.4)(0,0.3){1}{\tok}

\multirput(1,-0.2 )(1.2, 0){4}{\es}
\multirput(1,0)(0,0.3){8}{\tokII}
\multirput(2.2,0)(0,0.3){1}{\tokII}
\multirput(2.2,0.3)(0,0.3){5}{\tokI}
\multirput(3.4,0)(0,0.3){2}{\tok}
\multirput(4.6,0)(0,0.3){1}{\tok}

%\psline{|-|}(2,1.3)(2,1.9)
%\psline{|-|}(2,-0.2)(2,1.3)
%
%\rput(1.32,1.65){$\delta-\xi$}
%\rput(1.35,.6){$s_2(\xi)$}

%\multirput(3.6,0.35)(0,0.3){1}{\tokI}
\end{pspicture}
\end{center}
\caption{Two matched $(k-1)$-components and a tall third stack, allowing $\A$ a win using (NS7).}\label{NS7}
\end{figure}
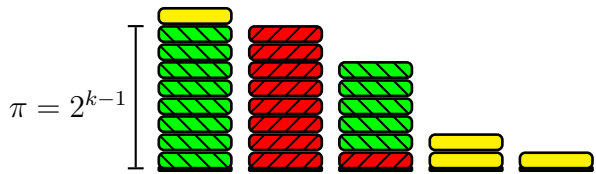

$\A$ can reach such a position if, after each move of $\B$
\begin{equation}\label{*} s_3(\xi)>\xi.\qquad  \end{equation}
%This follows since
Note that inequality~\eqref{*} holds at the beginning of the $\delta$-phase, as we assumed that $s_3=s_3(0)>0$.  It now remains to be seen whether the inequality can be maintained throughout. We may assume w.l.o.g. that inequality (\ref{Eq2}) does not hold for any $\xi$, as otherwise $\A$ wins by playing high. Now assume that  $s_3(\xi)-\xi>0$ and that $\A$ plays on $s_3$. Then, since $\B$ may also play on $s_3$, we have that
\begin{align}\label{Eq3}
s_3(\xi+1)-(\xi+1) \geqslant s_3(\xi)+1-(\xi+1)>0,
\end{align} unless $s_2(\xi)=s_3(\xi)$. This case could result in $s_2(\xi +1)=s_3(\xi)+1$ and $s_3(\xi +1)=s_3(\xi)$ (if $\B$ does not play on either stacks 2 or 3), so inequality (\ref{Eq3}) does not hold any longer. However, since
\eqref{Eq2} does not hold for any $\xi$ by assumption, we have
\begin{align}\label{Eq4}
s_3(\xi+1)-(\xi+1) = s_3(\xi) -(\xi+1)= s_2(\xi +1)-1-(\xi+1) \geqslant \pi-\delta-1>0
%1 \geqslant s_3(\xi +1)-\xi \geqslant s_2(\xi+1)-(\xi +1)\geqslant  \pi-\delta,
\end{align}
by \eqref{<delta}.
%In this case, $\A$ can build at least $\pi-1$ tokens on stack 3, since by assumption $\B$ builds $\pi$ tokens on stack 2 (and they start at the same height). If $s_3=\pi-1$ at the end of {\sc building}, since $n < 2^k-2$, we can apply (NS7) of Lemma~\ref{DSfact}.\\
%In case  $n<2^k$, if $\A$ builds $\pi$ tokens on stack 3, then there will be an odd number of $(k-1)$-components by the end of {\sc building} and $\A$ wins.\\
Thus, either \eqref{*} can be maintained throughout, or $\A$ can switch to playing high if \eqref{Eq1} holds at some point in the $\delta$-phase.

{\bf Case 2:} $1\leqslant n-p \leqslant 4$.

(i) If $\B$ played his first $p/2$ tokens on the second stack, $\A$ chooses $\pi=p/2=2^{k-2}$, so $\delta = n-\pi=n-2^{k-2}$. Therefore, $2^{k-2} < \delta \leqslant 2^{k-1}$, so $k(\delta)=k(n)-1$.  Note that since $n>12$, we have that $k(n)\geqslant 4$, and by induction hypothesis, $\bn(2\delta, 5)$ is winning for $\A$ since $3\leqslant k(\delta)= k(n)-1<k(n)$. $\A$ can apply her winning strategy on top of the matched stacks in the $\delta$-phase because
 $$ x_1+\cdots +x_5 =2n-2\pi =2n-p\leqslant n+4\leqslant p+8\leqslant 4\cdot2^{k-2}=4\pi,$$
so the contraposition of Lemma~\ref{DS8} applies and this strategy leads to a win for $\A$. \\

(ii) If $\B$ did not play his first  $p/2$ tokens on the second stack, then $\A$ goes on playing high on the tallest stack until the end of the $\pi$-phase, with $\pi=p$, and then plays the $\delta$-phase according to Case 1 (i) or (iii), assuring her win.
\end{proof}

\section{Discussion}

We have shown that in the case of three stacks, $\B$ has a winning strategy for $n=2^k-2$, which is no longer true for five stacks. With just three stacks, $\A$ does not have much wiggle room, and $\B$ can force a win, but with five stacks, $\A$ gains enough of an advantage in being able to play low. The proof of Lemma~\ref{2^k-2} can most likely be extended to more stacks, but in the proof of the main result, the cases where $\A$ uses a winning strategy for a smaller game on top of two stacks of size $2^i$ for some $i$ depends on verification by computer that $\A$ has winning strategies for a finite number of initial cases. The same would be true for any odd number of stacks $\ell >5$, with the number of initial cases increasing as the number of stacks increases. The conditions of Lemma~\ref{DS8} can be used to precisely define the number of initial games that are needed to use the induction argument. We do not currently have a general argument to prove that $\A$ can win these initial games but have found manual proofs for several values of $\ell$.  For many of $\A$'s winning strategies that we have checked it suffices for her to respond to $\B$'s defense attempts by playing `high or low', but we have also encountered cases where such strategies fail, where $\A$ still wins, but only by departing from `high or low' play. Conjecture~\ref{p1winsoften} suggests that $\B$ rarely wins for the interesting cases of {\sc building nim} ($n$ even and  $\ell$ odd), notably fitting the result by Singmaster \cite{Singm} that almost all games are first player wins.

In the process of our computer explorations we have also computed Grundy values for all strict {\sc building} positions for odd numbers of stacks $5\leqslant \ell \leqslant 19$ and an even number of tokens $\ell+3 < n\leqslant 34$. The Grundy function takes only the values $0, 1$ or $2$. More specifically, $\A$ moves from positions with Grundy value $0, 1$ or $2$, and $\B$ moves from positions with Grundy values $0$ or $1$. This gives rise to the following questions:
\begin{enumerate}
\item Does this observation hold in general?
\item Does this observation provide an answer to whether $\B$ only moves from Grundy value 0 in optimal play?
\end{enumerate}
We note that if the number of tokens is greater than the number of heaps, then the $\P$-positions of normal play {\sc building nim} are the same as those of the misere variation (a player who cannot move wins). Indeed, the $\P$-positions of Nim and mis\`ere Nim are the same, provided there is one heap of size at least $2$.

\end{document}